\documentclass[11pt]{article}

\usepackage{amsmath,amssymb,epsfig,theorem,fullpage,boxedminipage}
\usepackage{url}
\usepackage{setspace}
\usepackage[ruled, noend]{algorithm2e}
\usepackage{appendix}
\usepackage[T1]{fontenc}

\newtheorem{theorem}{Theorem}
\newtheorem{definition}{Definition}

\newtheorem{claim}{Claim}

\newtheorem{observation}{Observation}
\newtheorem{lemma}[theorem]{Lemma}

\newtheorem{corollary}[theorem]{Corollary}

\newcommand{\qed}{$\square$}

\newcommand{\MSVV}{MSVV}

\newcommand{\eat}[1]{}
\newcommand{\greedy}{{\sc Greedy}}
\newcommand{\ranking}{{\sc Ranking}}
\newcommand{\pgreedy}{{\sc Perturbed-Greedy}}
\newenvironment{proof}{\noindent{\em Proof:}}{\hfill \qed \medskip}

\title{Online Vertex-Weighted Bipartite Matching and\\ Single-bid Budgeted Allocations}

\author{Gagan Aggarwal\thanks{Google Inc., Mountain View. Email: \texttt{gagana@google.com}} \and Gagan Goel\thanks{Georgia Institute of Technology. Email: \texttt{gagang@cc.gatech.edu}} \and Chinmay Karande\thanks{Georgia Institute of Technology. Email: \texttt{ckarande@cc.gatech.edu}. Work done while visiting Google.} \and Aranyak Mehta\thanks{Google Inc., Mountain View. Email:  \texttt{aranyak@google.com}}}

\date{}

\begin{document}
\maketitle
\thispagestyle{empty}
\begin{abstract}
We study the following vertex-weighted online bipartite matching
problem: $G(U, V, E)$ is a bipartite graph. The vertices in $U$ have
weights and are known ahead of time, while the vertices in $V$ arrive
online in an arbitrary order and have to be matched upon arrival. The
goal is to maximize the sum of weights of the matched vertices in
$U$. When all the weights are equal, this reduces to the classic
\emph{online bipartite matching} problem for which Karp, Vazirani and
Vazirani gave an optimal $\left(1-\frac{1}{e}\right)$-competitive
algorithm in their seminal work~\cite{KVV90}.

Our main result is an optimal $\left(1-\frac{1}{e}\right)$-competitive
randomized algorithm for general vertex weights. We use \emph{random
  perturbations} of weights by appropriately chosen multiplicative
factors. Our solution constitutes the first known generalization of
the algorithm in~\cite{KVV90} in this model and
provides new insights into the role of randomization in online
allocation problems. It also effectively solves the problem of
\emph{online budgeted allocations} \cite{MSVV05} in the case when an
agent makes the same bid for any desired item, even if the bid is
comparable to his budget - complementing the results of \cite{MSVV05,
  BJN07} which apply when the bids are much smaller than the budgets.
\end{abstract}

\newpage
\setcounter{page}{1}

\section{Introduction}
\label{section:intro}
Online bipartite matching is a fundamental problem with numerous
applications such as matching candidates to jobs, ads to advertisers,
or boys to girls. A canonical result in online bipartite matching is
due to Karp, Vazirani and Vazirani~\cite{KVV90}, who gave an optimal
online algorithm for the unweighted case to maximize the {\em size} of
the matching. In their model, we are given a
bipartite graph $G(U, V, E)$. The vertices in $U$ are known ahead of
time, while the vertices in $V$ arrive one at a time online in an
arbitrary order. When a vertex in $V$ arrives, the edges incident to
it are revealed and it can be matched to a neighboring vertex in $U$
that has not already been matched. A match once made cannot be
revoked. The goal is to maximize the number of matched vertices. 

However, in many real world scenarios, the value received from
matching a vertex might be different for different vertices: (1)
Advertisers in online display ad-campaigns are willing to pay a fixed
amount every time their graphic ad is shown on a website. By
specifying their targeting criteria, they can choose the set of
websites they are interested in. Each impression of an ad can be
thought of as matching the impression to the advertiser, collecting
revenue equal to the advertiser's bid. (2) Consider the sale of an
inventory of items such as cars. Buyers arrive in an online manner
looking to purchase one out of a specified set of items they are
interested in. The sale of an item generates revenue equal to the
price of the item. The goal in both these cases is to maximize the
total revenue.  With this background, we consider the following
problem:\\

\noindent\textsc{\textbf{Online vertex-weighted bipartite matching:}} The input
instance is a bipartite graph\\ $G(U, V, E, \{b_u\}_{u\in U})$, with
the vertices in $U$ and their weights $b_u$ known ahead of
time. Vertices in $V$ arrive one at a time, online, revealing their
incident edges. An arriving vertex can be matched to an unmatched
neighbor upon arrival. Matches once made cannot be revoked later and a vertex
left unmatched upon arrival cannot be matched later. The goal is to
maximize the sum of the weights of the matched vertices in $U$.\\

\noindent {\bf Connection to the online budgeted allocation problem:}
Apart from being a natural generalization of the online bipartite
matching problem, our vertex-weighted matching problem is closely
related to an important class of online problems. Mehta \textit{et al}
\cite{MSVV05} considered the following online version of maximum
budgeted allocation problem \cite{GKP01, LLN01} to model sponsored
search auctions: We have $n$ agents and $m$ items. Each agent $i$
specifies a monetary budget $B_i$ and a bid $b_{ij}$ for each item
$j$. Items arrive online, and must be immediately allocated to an
agent. If a set $S$ of items is allocated to agent $i$, then the agent
pays the minimum of $B_i$ and $\sum_{j\in S}b_{ij}$. The objective is
to maximize the total revenue of the algorithm. An important and
unsolved restricted case of this problem is when all the non-zero bids
of an agent are equal, \textit{i.e.} $b_{ij} = b_i$ or 0 for all
$j$. This case reduces to our vertex-weighted matching problem (For a
proof, refer to Appendix \ref{app7}).

For the general online budgeted allocation problem, no factor better
than $\frac{1}{2}$ (achieved by a simple deterministic greedy
algorithm \cite{LLN01}) is yet known. The best known lower bound
stands at $1-\frac{1}{e}$ due to the hardness result in \cite{KVV90}
for the case when all bids and budgets are equal to 1 - which is
equivalent to the unweighted online matching problem. The \emph{small
  bids} case - where $b_{ij} \ll B_i$ for all $i$ and $j$ - was solved
by \cite{MSVV05, BJN07} achieving the optimal $1-\frac{1}{e}$
deterministic competitive ratio. It was believed that handling
\emph{large bids} requires the use of randomization, as in
\cite{KVV90}. In particular, many attempts \cite{KV07, BM08, GM08} had
been made to simplify the analysis of the randomized algorithm in
\cite{KVV90}, but no generalization had been achieved.

Our solution to the vertex-weighted matching problem is a significant
step in this direction. Our algorithm generalizes that of \cite{KVV90}
and provides new insights into the role of randomization in these
solutions, as outlined in Section \ref{section:overview}. Finally,
our algorithm has interesting connections to the solution of
\cite{MSVV05} for the \emph{small bids} case - despite the fact that
the vertex-weighted matching problem is neither harder nor easier than
the \emph{small bids} case. This strongly suggests a possible unified
approach to the unrestricted online budgeted allocation problem. See
Section \ref{sec:implications} for details.

\subsection{Overview of the Result}
\label{section:overview}

\noindent{\bf Solution to the unweighted case}: To describe our result, it is instructive to start at the unweighted case ($b_u = 1$ for all $u\in U$) and study its solution by \cite{KVV90}. Two natural approaches that match each arriving $v \in V$ to the an unmatched neighbor in $U$ chosen (a) arbitrarily and (b) randomly, both fail to achieve competitive ratio better than $\frac{1}{2}$. Their solution is an elegant randomized algorithm called \ranking~that works as follows: it begins by picking a \emph{uniformly random permutation} of the vertices in $U$ (called the ``ranking'' of the vertices). Then, as a vertex in $V$ arrives, it is matched to the highest-ranked unmatched neighbor. Surprisingly, this idea of using correlated randomness for all the arriving vertices achieves the optimal competitive ratio of $1-\frac{1}{e}$.\\

How do we generalize \ranking~in presence of unrestricted weights $b_u$? The natural \greedy~algorithm which matches an arriving vertex to the highest-weighted unmatched neighbor, achieves a competitive ratio of $\frac{1}{2}$ (see
Appendix \ref{app3} for a proof). No deterministic algorithm can do
better. While the optimality of \ranking~for unweighted matching suggests choosing random ranking permutations of $U$, \ranking~itself can do as badly as factor $\frac{1}{n}$ for some weighted instances.

The main challenge in solving this problem
is that a good algorithm must follow very different strategies
depending on the weights in the input instance. \greedy~and
\ranking~are both suboptimal for this problem, but both have ideas
which are essential to its solution. In particular, they perform well
on distinct classes of inputs, namely, \greedy~on highly skewed
weights and \ranking~on equal weights. The following observation about \ranking~helps us bridge the gap between these two approaches: Suppose we perturb each weight $b_u$ identically and independently and then sort the vertices in the order of decreasing perturbed weights. When all the weights are equal, the resulting order happens to be a uniformly random permutation of $U$ and thus, \ranking~on unweighted instances can be thought of as \greedy~on perturbed weights! We use this insight to construct our solution to the vertex-weighted matching problem. While the nature of perturbation used did not matter in the above discussion, we need a very specific perturbation procedure for general vertex-weights.

Our algorithm is defined below:

\begin{algorithm}[H]
\caption{{\sc Perturbed-Greedy}} For each $u \in U$, pick a number
$x_u$ uniformly at random from $[0, 1]$.\\
Define the function $\psi(x) := 1 - e^{-(1-x)}$.\\
\ForEach{arriving $v \in V$} { Match $v$ to the unmatched neighbor
$u \in U$ with the highest value of $b_u \psi(x_u)$. Break ties
consistently, say by vertex id. }
\end{algorithm}

\noindent{\bf Remarks}: It is not obvious, and indeed is remarkable in
our opinion, that it suffices to perturb each weight $b_u$ completely
independently of other weights. In Appendix \ref{app6}, we provide
intuition as to why such is the case. Also, the particular form of the
function $\psi$ is not a pre-conceived choice, but rather an artifact
of our analysis. This combined with the discussion in Section
\ref{sec:implications} seems to suggest that $\psi$ is the `right'
perturbation function. We note that we can also choose the function
$\psi(x)$ to be $1-e^{-x}$, which keeps the algorithm and results
unchanged. Finally, we note that the multipliers $y_u = \psi(x_u)$ are
distributed according to the density function $f(y) = \frac{1}{1-y}$
for $y \in \left[0, 1-\frac{1}{e}\right]$. Therefore, we could have
equivalently stated our algorithm as: For each $u\in U$, choose a
random multiplier $y_u \in \left[0, 1-\frac{1}{e}\right]$ from the
above distribution, and use $b_uy_u$ as the perturbed weight.

Our main result is the following theorem. The second part of
the theorem follows from the optimality of \ranking~for unweighted
matching \cite{KVV90}.
\begin{theorem}
\label{thm:main}
\pgreedy~achieves a competitive ratio of $1-1/e$ for the
vertex-weighted online bipartite matching problem. No (randomized)
algorithm has a better competitive ratio.
\end{theorem}

In addition to the basic idea (from the proof of \ranking) of charging
unmatched vertices in some probabilistic events to matched vertices in
other events, our analysis needs to handle the new complexity
introduced due to the weights on vertices. At a very high level, just
like the algorithm, our analysis also manages to pull together the
essence of the analyses of both \greedy~and \ranking.

\subsection{Implications of the Result}
\label{sec:implications}

\noindent{\bf Finding the optimal distribution over permutations of
  $U$:} Since \pgreedy~also chooses ranking orders through
randomization, we can interpret it as a non-uniform \ranking, where it
chooses permutations of $U$ from the `optimal' distribution. But we
could have posed the following question, without the knowledge of our
algorithm: How do we find an optimal non-uniform distribution over
permutations of $U$? As a start, let us consider the case of $2\times
2$ graphs. By exhaustive search over all $2\times 2$ graphs, we can
figure out the best \ranking~like algorithm for $2\times 2$ graphs
(Figure \ref{fig:twobytwo} in Appendix \ref{app9} shows the only two
potentially `hard' instances in $2\times 2$ graphs). This algorithm
picks the permutation $(u_1, u_2)$ with probability
$\tfrac{\alpha}{1+\alpha}$ and the permutation $(u_2, u_1)$ with
probability $\tfrac{1}{1+\alpha}$ (where $\alpha = b_{u_1} /
b_{u_2}$), and then proceeds to match to the highest neighbor. This
algorithm gives a factor of $\frac{\alpha^2 + \alpha + 1}{(\alpha +
  1)^2}$, which is minimized at $\alpha = 1$, giving a factor of $3/4$
(in which case the algorithm is simply the same as \ranking).

An attempt to generalize this idea to larger graphs fails due to a
blow-up in complexity.  In general, we need a probability variable
$p_\sigma$ for every permutation $\sigma$ of $U$. The expected weight
of the matching produced by the algorithm on a graph $G$, is a linear
expression $\mathrm{ALG}_G(p_{\sigma_1}, p_{\sigma_2}, ...)$. Thus,
the optimal distribution over permutations is given by the optimal
solution of a linear program in the $p_\sigma$ variables. But this LP
has exponentially many variables (one per permutation) and constraints
(one per ``canonical graph instance''). Therefore, our algorithm can
be thought of as solving this extremely large LP through a very simple
process.\\

\noindent \textbf{General capacities / Matching $u\in U$ multiple
  times}: Consider the following generalization of the online
vertex-weighted bipartite matching problem: Apart from a weight $b_u$,
each vertex $u \in U$ has a capacity $c_u$ such that $u$ can be
matched to \emph{at most} $c_u$ vertices in $V$. The capacities allow
us to better model `budgets' in many practical situations,
\textit{e.g.}, in online advertising. Our algorithm easily handles
general capacities: For each $u\in U$, make $c_u$ copies of $u$ and
solve the resulting instance with unit capacities: It is easy to
verify that the solution is $\left(1-\frac{1}{e}\right)$-approximate
in expectation for the original problem with capacities.\\

\noindent{\bf Online budgeted allocation :- The \emph{single bids}
  case vs. the \emph{small bids} case}: As noted earlier and proved in
Appendix \ref{app7}, the special case of the online budgeted
allocation problem with all the non-zero bids of an agent being equal
($b_{ij} = b_i$ or 0), reduces to our vertex-weighted matching
problem. Since each agent provides a single bid value for all items,
let us call this restriction the \emph{single bids} case.

\begin{corollary}
\pgreedy~achieves a competitive ratio of $1-1/e$ for the \emph{single
  bids} case of the online budgeted allocation problem.
\end{corollary}

Note that the \emph{small bids} case ($b_{ij} \ll B_i$) studied in
\cite{MSVV05, BJN07} does not reduce to or from the \emph{single bids}
case. Yet, as it turns out, \pgreedy~is equivalent to the algorithm of
\cite{MSVV05} - let us call it MSVV - on instances that belong to the
intersection of the two cases. When every agent has a \emph{single
  small bid} value, the problem corresponds to vertex-weighted
matching with large capacities $c_u$ for every vertex $u$. Recall that
we handle capacities on $u\in U$ by making $c_u$ copies $u_1, u_2,
..., u_{c_u}$ of $u$. For each of these copies, we choose a random
$x_{u_i} \in [0,1]$ uniformly and independently. In expectation, the
$x_{u_i}$'s are uniformly distributed in the interval $[0,1]$. Also
observe that \pgreedy~will match $u_1, u_2, ..., u_{c_u}$ in the
increasing order of $x_{u_i}$'s, if at all. Therefore, at any point in
the algorithm, if $u_i$ is the unmatched copy of $u$ with smallest
$x_{u_i}$ (and consequently highest multiplier $\psi(x_{u_i})$) then
$x_{u_i}$ is in expectation equal to the fraction of the capacity
$c_u$ used up at that point. But \MSVV~uses exactly the scaling factor
$\psi(T)$ where $T$ is the fraction of spent budget at any point. We
conclude that in expectation, \pgreedy~tends to \MSVV~as the
capacities grow large, in the single small bids case.

It is important to see that this phenomenon is not merely a
consequence of the common choice of function $\psi$. In fact, the
function $\psi$ is not a matter of choice at all - it is a by-product
of both analyses (Refer to the remark at the end of Section
\ref{sec:main-proof}). The fact that it happens to be the exact same
function seems to suggest that $\psi$ is the `right'
function. Moreover, the analyses of the two algorithms do not imply
one-another. Our variables are about expected gains and losses over a
probability space, while the algorithm in~\cite{MSVV05} is purely
deterministic.

This smooth `interface' between the seemingly unrelated \emph{single bids} and \emph{small bids} cases hints towards the existence of a unified solution to the general online budgeted allocation problem.

\subsection{Other Related Work}\label{related_work}
Our problem is a special case of online bipartite matching with edge
weights, which has been studied extensively in the literature. With
general edge weights and vertices arriving in adversarial order, every
algorithm can be arbitrarily bad (see Appendix~\ref{app4}). There are two ways to
get around this hardness: (a) assume that vertices arrive in a random
order, and/or (b) assume some restriction on the edge weights.

When the vertices arrive in random order, it corresponds to a
generalization of the {\em secretary} problem to transversal
matroids~\cite{BIK07}. Dimitrov and Plaxton~\cite{DP08} study a
special case where the weight of an edge $(u,v)$ depends only on the
vertex $v$ -- this is similar to the problem we study, except that it
assumes a random arrival model (and assumes vertex weights on the {\em
  online} side). Korula and Pal~\cite{KP09} give an
$\tfrac{1}{8}$-competitive algorithm for the problem with general edge
weights and for the general {\em secretary} problem on transversal
matroids.

If one does not assume random arrival order, every algorithm
can be arbitrarily bad with general edge weights or even with weights
on arriving vertices. ~\cite{KP93} introduce the assumption of edge
weights coming from a metric space and give an optimal deterministic
algorithm with a competitive factor of $\tfrac{1}{3}$. As far as we
know, no better randomized algorithm is known for this problem.

Finally, there has been other recent work \cite{DH09,GM08,FMMM09}, although
not directly related to our results, which study online bipartite
matching and budgeted allocations in stochastic arrival settings.\\

\noindent{\bf Roadmap:} The rest of the paper is structured as follows: In
Section~\ref{sec:prelim} we set up the preliminaries and provide a
warm up analysis of a proof of \ranking~in the unweighted special
case. Section~\ref{sec:main-proof} contains the proof of
Theorem~\ref{thm:main}.

\section{Preliminaries}
\label{sec:prelim}
\subsection{Problem Statement}
\label{section:statement}
Consider an undirected bipartite graph $G(U,V,E)$. The vertices of
$U$, which we will refer to as the \emph{offline} side, are known from
the start. We are also given a weight $b_u$ for each vertex $u \in U$.
The vertices of $V$, referred to as the \emph{online} side, arrive one
at a time (in an arbitrary order). When a vertex $v$ arrives, all the
edges incident to it are revealed, and at this point, the vertex $v$
can be matched to one of its unmatched neighbors (irrevocably) or left
permanently unmatched. The goal is to maximize the sum of the weights
of matched vertices in $U$.

Let permutation $\pi$ represent the arrival order of vertices in $V$
and let $M$ be the subset of matched vertices of $U$ at the end. Then
for the input $(G, \pi)$, the gain of the algorithm, denoted by $ALG(G,
\pi)$, is $\sum_{u \in M}{b_u}$.

We use competitive analysis to analyze the performance of an
algorithm.  Let $M^*(G)$ be an optimal (offline) matching, i.e. one
that maximizes the total gain for $G$ (note that the optimal matching
depends only on $G$, and is independent of $\pi$), and let
$\mathrm{OPT}(G)$ be the total gain achieved by $M^*(G)$. Then the
competitive ratio of an algorithm is $\min_{G, \pi}
\tfrac{\mathrm{ALG}(G, \pi)}{\mathrm{OPT}(G)}$. Our goal is to devise
an online algorithm with a high competitive ratio.

\begin{definition}[$M^*(G)$] 
For a given $G$, we will fix a particular optimal matching, and refer
to it as the optimal offline matching $M^*(G)$.
\end{definition}
 
\begin{definition}[$u^*$]
Given a $G$, its optimal offline matching $M^*(G)$ and a $u \in U$
that is matched in $M^*(G)$, we define $u^* \in V$ as its partner in $M^*(G)$.
\end{definition}

\subsection{Warm-up: Analysis of {\sc Ranking} for Unweighted Online Bipartite Matching}
\label{section:KVV}

Recall that online bipartite matching is a special case of our problem
in which the weight of each vertex is $1$, i.e. $b_u = 1$ for all $u
\in U$.  \cite{KVV90} gave an elegant randomized algorithm for this
problem and showed that it achieves a competitive ratio of $(1-1/e)$
in expectation. In this section, we will re-prove this classical result
as a warm-up for the proof of the main result. The following proof is
based on those presented by \cite{BM08,GM08} previously.

\vspace{0.1in}

\begin{algorithm}[H]
\caption{{\sc Ranking}}
Choose a random permutation $\sigma$ of $U$ uniformly from the space of all permutations.\\
\ForEach{arriving $v \in V$}
{
	Match $v$ to the unmatched neighbor in $u$ which appears earliest in $\sigma$.\\
}
\end{algorithm}

\begin{theorem}[~\cite{KVV90}]
\label{thm:kvv}
In expectation, the competitive ratio of {\sc Ranking} is at least $1-\frac{1}{e}$.
\end{theorem}

In this warm-up exercise, we will simplify the analysis by making the following assumptions: $|U| = |V| = n$ and $G$ has a perfect matching. These two assumptions imply that $\mathrm{OPT} = n$ and that the optimal matching $M^*(G)$ is a perfect matching.

For any permutation $\sigma$, let {\sc Ranking}$(\sigma$) denote the
matching produced by {\sc Ranking} when the randomly chosen
permutation happens to be $\sigma$.
For a permutation $\sigma = (u_1, u_2, ..., u_n)$ of $U$, we say that
a vertex $u = u_t$ has rank $\sigma(u) = t$. 
Consider the random variable 
$$y_{\sigma, i}\ =\ \left\lbrace
\begin{tabular}{c c c} 
1 & & If the vertex at rank $i$ in $\sigma$ is matched by {\sc Ranking}$(\sigma)$.\\ 
0 & & Otherwise\\ 
\end{tabular}
\right.$$

\begin{definition}[$Q_t$, $R_t$]
$Q_t$ is defined as the set of all occurrences of matched vertices
  in the probability space.
$$Q_t = \{\ (\sigma, t)\ :\ y_{\sigma,t} = 1\ \}$$
Similarly, $R_t$ is defined as the set of all occurrences of unmatched vertices
in the probability space.
$$R_t =
\{\ (\sigma, t)\ :\ y_{\sigma,t} = 0\ \}$$
\end{definition}

Let $x_t$ be the probability that the vertex at rank $t$ in $\sigma$
is matched in {\sc Ranking}$(\sigma)$, over the random choice of
permutation $\sigma$. Then, $x_t\ =\ \frac{|Q_t|}{n!}$ and
$1-x_t\ =\ \frac{|R_t|}{n!}$. The expected gain of the algorithm is
$\mathrm{ALG}_{G, \pi} = \sum_tx_t$.

\begin{definition}[$\sigma_u^i$]
For any $\sigma$, let
$\sigma_u^i$ be the permutation obtained by removing $u$ from $\sigma$
and inserting it back into $\sigma$ at position $i$.
\end{definition}

\begin{lemma}
\label{lemma1}
If the vertex $u$ at rank $t$ in $\sigma$ is unmatched by {\sc
  Ranking}($\sigma$), then for every $1 \leq i \leq n$, $u^*$ is matched in
{\sc Ranking}($\sigma_u^i$) to a vertex $u'$ such that $\sigma_u^i(u') \leq
t$.
\end{lemma}
\begin{proof}
Refer to Lemma~\ref{lemma2} in the analysis of
\pgreedy~for the proof of a more general version of this statement.
\end{proof}

In other words, for every vertex that remains unmatched in some event
in the probability space, there are many matched vertices in many
different events in the space. In the remaining part of this section,
we quantify this effect by bounding $1-x_t$, which is the probability
that the vertex at rank $t$ in $\sigma$ (chosen randomly by {\sc
  Ranking}) is unmatched, in terms of some of the $x_t$s.

\begin{definition}[Charging map $f(\sigma, t)$]
$f$ is a map from bad events (where vertices remain unmatched) to good
  events (where vertices get matched).  For each $(\sigma,t) \in
  R_t$, $$f(\sigma, t)\ =\ \{ (\sigma_u^i, s)\ :\ \mbox{$1\leq i\leq
    n$, $\sigma(u) = t$ and {\sc Ranking}$(\sigma_u^i)$ matches $u^*$
    to $u'$ where $\sigma_u^i(u') = s$}\}$$
\end{definition}

In other words, let $u$ be the vertex at rank $t$ in $\sigma$. Then
$f(\sigma,t)$ contains all $(\sigma',s)$, such that $\sigma'$ can be
obtained from $\sigma$ by moving $u$ to some position and $s$ is the
rank of the vertex to which $u^*$, the optimal partner of $u$, is
matched in $\sigma'$.

For every $(\sigma, t) \in R_t$, $(\pi, s) \in f(\sigma, t)$ implies
$y_{\pi, s} = 1$ for some $s \leq t$. Therefore, $$\bigcup_{(\sigma,
  t)\in R_t}{f(\sigma, t)}\ \subseteq\ \bigcup_{s\leq t}{Q_s}$$

\begin{claim}
\label{claim1}
If $(\rho, s) \in f(\sigma, t)$ and $(\rho, s) \in
f(\overline{\sigma}, t)$, then $\sigma = \overline{\sigma}$.
\end{claim}
\begin{proof}
Let $u'$ be the vertex in $\rho$ at rank $s$. Let $u^*$ be the vertex
to which $u'$ is matched by {\sc Ranking}. Then it is clear
from the definition of the map $f$ that $\rho = \sigma_u^{\rho(u)} =
\overline{\sigma}_u^{\rho(u)}$, implying $\sigma=\overline{\sigma}$.
\end{proof}

The claim proves that for a fixed $t$, the set-values $f(\sigma,
t)$ are disjoint for different $\sigma$. Therefore, 
$$ 1-x_t
\ =\ \frac{|R_t|}{n!}\ =\ \frac{1}{n}\cdot\frac{\left|\bigcup_{(\sigma,
    t)\in R_t}{f(\sigma,
    t)}\right|}{n!}\ \leq\ \frac{1}{n}\cdot\frac{\left|\bigcup_{s\leq
    t}{Q_s}\right|}{n!} \ =\ \frac{1}{n}\sum_{s\leq
  t}{\frac{|Q_s|}{n!}}\ =\ \frac{\sum_{s\leq t}{x_s}}{n}$$

Therefore, the probabilities $x_t$'s obey the equation
$1-x_t\ \leq\ \frac{1}{n}\sum_{s\leq t}{x_s}$ for all $t$. Since any vertex with
rank 1 in any of the random permutations will be matched, $x_1 =
1$. One can make simple arguments \cite{KVV90, BM08, GM08} to prove that
under these conditions, $\mathrm{ALG}_{G,\pi} = \sum_t{x_t} \geq
\left(1-\frac{1}{e}\right)n = \left(1-\frac{1}{e}\right) OPT$,
thereby proving Theorem~\ref{thm:kvv}.

\section{Proof Of Theorem~\ref{thm:main}}
\label{sec:main-proof}

In this section, we will assume that $|U| = |V| = n$ and that $G$ has
a perfect matching. In Appendix \ref{app1} we will show how this
assumption can be removed.

Recall that our algorithm works as follows: For each $u \in U$, let
$\sigma(u)$ be a number picked uniformly at random from $[0,1]$ (and
independent of other vertices) Now, when the next vertex $v \in V$
arrives, match it to the available neighbor $u$ with the maximum value
of $b_u\psi(\sigma(u))$, where $\psi(x) := 1 - e^{-(1-x)}$.

For ease of exposition, we will prove our result for a discrete
version of this algorithm. For every $u\in U$ we will choose a random
integer $\sigma(u)$ uniformly from $\{1, ..., k\}$ where $k$ is the
parameter of discretization. We will also replace the
function $\psi(x)$ by its discrete version $\psi(i) = 1 -
\left(1-\frac{1}{k}\right)^{-(k-i+1)}$. The discrete version of our
algorithm also matches each incoming vertex $v \in V$ to the available
neighbor $u$ with the maximum value of $b_u\psi(\sigma(u))$. Notice
that $\psi$ is a decreasing function, so $\psi(s) \geq \psi(t)$ if $s
\leq t$. As $k \rightarrow \infty$, the discrete version tends to our
original algorithm.\\

We begin with some definitions, followed by an overview of the proof.

We will denote by $\sigma \in
[k]^n$, the set of these random choices. We will say that $u$ is at
\emph{position} $t$ in $\sigma$ if $\sigma(u) = t$. As a matter of
notation, we will say that position $s$ is \emph{lower} (resp. higher)
than $t$ if $s \leq t$ (resp. $s \geq t$).

\begin{definition}[$u$ is matched in $\sigma$]
We say that $u$ is matched in $\sigma$ if our algorithm matches it
when the overall choice of random positions happens to be
$\sigma$. 
\end{definition}

Let $y_{\sigma,t}$ be the indicator variable denoting that the vertex
at position $t$ is matched in $\sigma$. 

\begin{definition}[$Q_t$, $R_t$]
$Q_t$ is defined as the set of all occurrences of matched vertices in the probability space.
$$Q_t = \{(\sigma, t, u)\ :\ \mbox{$\sigma(u) = t$ and
  $y_{\sigma,t} = 1$}\}$$
Similarly, $R_t$ is defined as the set of all occurrences of unmatched vertices in the probability space.
$$R_t = \{(\sigma, t,
u)\ :\ \mbox{$\sigma(u) = t$ and $y_{\sigma,t} = 0$}\}$$
\end{definition}

Let $x_t$ be the \emph{expected gain} at $t$, over the random choice
of $\sigma$. Then,
\begin{equation}
\label{eq9}
x_t = \frac{\sum_{(\sigma, t, u)\in Q_t}{b_u}}{k^n}
\end{equation}

The expected gain of the algorithm is $\mathrm{ALG}_{G, \pi} =
\sum_tx_t$. Also note that the \emph{optimal gain} at any position
$t$ is $B = \frac{\mathrm{OPT}(G)}{k}$ since each vertex in $U$
appears at position $t$ with probability $1/k$ and is matched in the
optimal matching. Therefore,

\begin{equation}
\label{eq8}
B-x_t = \frac{\sum_{(\sigma, t, u)\in R_t}{b_u}}{k^n}
\end{equation}

\begin{definition}[$\sigma_u^i$]
For any $\sigma$, $\sigma_u^i \in [k]^n$ is obtained from $\sigma$ by
changing the position of $u$ to $i$, i.e. $\sigma_u^i(u) = i$ and
$\sigma_u^i(u') = \sigma(u')$ for all $u' \neq u$.
\end{definition}

\begin{observation}
For all $(\sigma, t, u) \in R_t$ and $1 \leq i \leq k$, our algorithm
matches $u^*$ to some $u' \in U$ in $\sigma_u^i$.
\end{observation}
The above observation follows from Lemma~\ref{lemma2}. We'll use it
to define a map from bad events to good events as follows.

\begin{definition}[Charging Map $f(\sigma, t, u)$]
\label{def:map}
For every $(\sigma, t, u) \in R_t$, define the set-valued
map $$f(\sigma, t, u) = \{(\sigma_u^i, s, u')\ :\ \mbox{$1\leq i\leq
  k$, and the algorithm matches $u^*$ to $u'$ in $\sigma_u^i$ where
  $\sigma_u^i(u') = s$}\}$$
\end{definition}

\begin{observation}
\label{obs1}
If $(\rho, s, u') \in f(\sigma, t, u)$, then $(\rho, s, u') \in Q_s$.
\end{observation}

Now we are ready to give an overview of the proof.

\subsection*{Overview of the proof}
The key idea in the analysis of {\sc Ranking} in Section
\ref{section:KVV} was that we can bound the number of occurrences of
unmatched vertices - the \emph{bad} events - in the entire probability
space by a careful count of the matched vertices - the \emph{good}
events. The charging map $f$ defined above is an attempt to do this.
We'll show in Lemma~\ref{lemma2} that if $(\sigma_u^i, s, u') \in
f(\sigma, t, u)$, then the scaled (by $\psi$) gain due to $u'$ in
$\sigma_u^i$ is no less than the scaled loss due to $u$ in
$\sigma$. However, $s$ may be higher or lower than $t$, unlike {\sc
  Ranking} where $s \leq t$. This implies that the bound is in terms
of events in $\bigcup_s Q_s$, $1 \leq s \leq k$, which is very weak (as
many of the events in the union are not used). 

One idea is to bound the sum of losses incurred at all positions,
thereby using almost all the events in $\bigcup_s Q_s$. However, if we do this,
then the charging map loses the disjointness property, i.e. if
$(\sigma, t, u) \in R_t$ and $(\sigma_u^i, i, u) \in R_i$ then $f$
value of both these occurrences is the same. Thus, each event in
$\bigcup_s Q_s$ gets charged several times (in fact a non-uniform number
of times), again making the bound weak. To this end, we introduce the idea of {\em marginal loss}~(Definition \ref{def:margin}), which helps us define a disjoint map and get a tight bound.

Next, we formalize the above.

\subsection*{Formal proof}

We begin by proving an analogue of Lemma~\ref{lemma1}.

\begin{lemma}
\label{lemma2}
If the vertex $u$ at position $t$ in $\sigma$ is unmatched by our algorithm, then for every $1 \leq i\leq k$, the algorithm matches $u^*$ in $\sigma_u^i$ to a vertex $u'$ such that $\psi(t)b_u \leq \psi\left(\sigma_u^i(u')\right)b_{u'}$.
\end{lemma}
\begin{proof} 
\noindent{\textbf{Case 1 ($i \geq t$)}}: Let $v_1, ..., v_n$ be the order of arrival of vertices in
$V$. Clearly, $v_1$ will see the same choice of neighbors in
$\sigma_u^i$ as in $\sigma$, except the fact that the position of $u$
is higher in $\sigma_u^i$ than in $\sigma$. Since we did not match
$v_1$ to $u$ in $\sigma$, $v_1$ will retain its match from $\sigma$
even in $\sigma_u^i$. Now assuming that $v_1, ..., v_l$ all match the
same vertex in $\sigma_u^i$ as they did in $\sigma$, $v_{l+1}$ will
see the same choice of neighbors in $\sigma_u^i$ as in $\sigma$ with
the exception of $u$. Since $v_{l+1}$ did not match $u$ in $\sigma$
either, it will retain the same neighbor in $\sigma_u^i$ and by
induction every vertex from $V$, specifically $u^*$ keeps the same
match in $\sigma_u^i$ as in $\sigma$. Since $\sigma(u') =
\sigma_u^i(u')$, we conclude $\psi(t)b_u \leq
\psi\left(\sigma_u^i(u')\right)b_{u'}$.\\

\noindent{\textbf{Case 2 ($i < t$)}}: For a vertex $v \in V$, let $m_\sigma(v)$ and $m_{\sigma_u^i}(v)$ be
the vertices to which $v$ is matched in $\sigma$ and $\sigma_u^i$
respectively, if such a match exists and null otherwise. Intuitively,
since $\psi(i) \geq \psi(t)$, the scaling factor of $b_u$ only
improves in this case, while that of any other vertex in $U$ remains
the same. Therefore, we can expect $u$ to be more likely to be matched
in $\sigma_u^i$ and the
$\psi\left(\sigma_u^i\left(m_{\sigma_u^i}(v)\right)\right)b_{m_{\sigma_u^i}(v)}
\geq
\psi\left(\sigma\left(m_{\sigma}(v)\right)\right)b_{m_{\sigma}(v)}$ to
hold for all $v \in V$. In fact, something more specific is true. The
symmetric difference of the two matchings produced by the algorithm
for $\sigma$ and $\sigma_u^i$ is exactly one path starting at $u$ that
looks like $(u,\ v_1,\ m_\sigma(v_1),\ v_2,\ m_\sigma(v_2),\ ...)$,
where $(v_1, v_2, ...)$ appear in their order of arrival. In what
follows we prove this formally.

Let $V' = \{ v\in V\ :\ m_\sigma(v) \neq m_{\sigma_u^i}(v)\}$ be the
set of vertices in $V$ with different matches in $\sigma$ and
$\sigma_u^i$. Index the members of $V'$ as $v_1, ..., v_l$ in the same
order as their arrival, \textit{i.e.} $v_1$ arrives the earliest. For
simplicity, let $u_j = m_\sigma(v_j)$ and $w_j = m_{\sigma_u^i}(v_j)$.

We assert that the following invariant holds for $2 \leq j \leq l$:
Both $u_j$ and $u_{j-1}$ are unmatched in $\sigma_u^i$ when $v_j$
arrives and $v_j$ matches $u_{j-1}$, \textit{i.e.} $w_j = u_{j-1}$.

For base case, observe that the choice of neighbors for $v_1$ in
$\sigma_u^i$ is the same as in $\sigma$, except $u$, which has moved
to a lower position. Since by definition $v_1$ does not match $u_1$ in
$\sigma_u^i$, $w_1 = u$. Now consider the situation when $v_2$
arrives. All the vertices arriving before $v_2$ - with the exception
of $v_1$ - have been matched to the same vertex in $\sigma_u^i$ as in
$\sigma$, and $v_1$ has matched to $u$, leaving $u_1$ yet
unmatched. Let $U_\sigma(v_2)$ and $U_{\sigma_u^i}(v_2)$ be the sets
of unmatched neighbors of $v_2$ in $\sigma$ and $\sigma_u^i$
respectively \emph{at the moment} when $v_2$ arrives. Then from above
arguments, $U_{\sigma_u^i}(v_2) = \left(U_\sigma(v_2)\cup
\{u_1\}\right)-\{u\}$. Since $u$ was unmatched in $\sigma$, $u_2 \neq
u$. Since $v_2 \in V'$, $w_2 \neq u_2$. This is only possible if $w_2
= u_1$. And hence the base case is true.

Now assume that the statement holds for $j-1$ and consider the arrival
of $v_j$. By induction hypothesis, $v_1$ has been matched to $u$ and
$v_2, .., v_{j-1}$ have been matched to $u_1, ..., u_{j-2}$
respectively. All the other vertices arriving before $v_j$ that are
not in $V'$ have been matched to the same vertex in $\sigma_u^i$ as in
$\sigma$. Therefore, $u_{j-1}$ is yet unmatched. Let $U_\sigma(v_j)$
and $U_{\sigma_u^i}(v_j)$ be the sets of unmatched neighbors of $v_j$
in $\sigma$ and $\sigma_u^i$ respectively at the moment when $v_j$
arrives. Then from above arguments, $U_{\sigma_u^i}(v_j) =
\left(U_\sigma(v_j)\cup \{u_{j-1}\}\right)-\{u\}$. Since $u$ was
unmatched in $\sigma$, $u_j \neq u$. Given that $w_j \neq u_j$, the
only possibility is $w_j = u_{j-1}$. Hence the proof of the inductive
statement is complete.

If $u^* \notin V'$ then $u' = m_{\sigma_u^i}(u^*) = m_\sigma(u^*)$ and
the statement of the lemma clearly holds since $\sigma(u') =
\sigma_u^i(u')$. If $u^* = v_1$, then $u' = u$ and
$\psi\left(\sigma_u^i(u')\right)b_{u'} = \psi(i)b_u \geq \psi(t)b_u$
since $i < t$. Now suppose $u^* = v_j$ for some $j \geq 2$. Then $u' =
u_{j-1}$ and by the invariant above,
\begin{eqnarray}
\label{eq11} \psi\left(\sigma_u^i(u')\right)b_{u'}\ = \ \psi\left(\sigma_u^i(u_{j-1})\right)b_{u_{j-1}} & \geq & \psi\left(\sigma_u^i(u_j)\right)b_{u_j}\\
\label{eq12} & = & \psi\left(\sigma(u_j)\right)b_{u_j}\\
\label{eq13} & \geq & \psi(t)b_u
\end{eqnarray}

Equation \eqref{eq11} follows from the fact that $u^* = v_j$ was matched in
$\sigma_u^i$ to $u_{j-1}$ when $u_j$ was also unmatched. The fact that
only $u$ changes its position between $\sigma$ and $\sigma_u^i$ leads us
to \eqref{eq12}. Finally, equation \eqref{eq13} follows from the fact
that $u^*$ was matched to $u_j$ in $\sigma$ when $u$ was also
unmatched.
\end{proof}


Using the above lemma, we get the following easy observation.
\begin{observation}
\label{obsk}
For all $(\sigma, t, u) \in R_t$, $1 \leq t \leq k$, $f(\sigma, t, u)$
contains $k$ values.
\end{observation}

\textbf{Remark}: As noted in the overview, although Lemma \ref{lemma2}
looks very similar to Lemma \ref{lemma1}, it is not sufficient to get
the result, since the good events pointed to by Lemma \ref{lemma2} are
scattered among all positions $1 \leq s \leq k$ -- in contrast to
Lemma \ref{lemma1}, which pointed to only lower positions $s\leq t$,
giving too weak a bound. We try to fix this by combining the losses
from all $R_t$. However we run into another difficulty in doing
so. While for any fixed $t$, the maps $f(\sigma, t, u)$ are disjoint
for all $(\sigma, t, u) \in R_t$, but the maps for two occurrences in
different $R_t$s may not be disjoint.
In fact, whenever some $u$ is unmatched in $\sigma$ at position
$t$, it will also remain unmatched in $\sigma_u^j$ for
$j > t$, and the sets $f(\sigma, t, u)$ and $f(\sigma_u^j, j, u)$ will
be exactly the same! This situation is depicted in Figure \ref{fig:margins} in Appendix \ref{app8}.

This absence of disjointness again renders the bound too weak. To fix this, we carefully select a subset of bad events
from $\bigcup_t{R_t}$ such that their set-functions are indeed
disjoint, while at the same time, the total gain/loss can be easily
expressed in terms of the bad events in this subset.
   
\begin{definition}[Marginal loss events $S_t$] 
\label{def:margin}
Let $S_t = \{(\sigma, t, u) \in R_t: (\sigma_u^{t-1},
t-1, u) \notin R_{t-1}\}$, where $R_0 = \emptyset$.
\end{definition}

Informally, $S_t$ consists of \emph{marginal} losses. If $u$ is
unmatched at position $t$ in $\sigma$, but matched at position $t-1$
in $\sigma_u^{t-1}$, then $(\sigma, t, u) \in S_t$ (See Figure
\ref{fig:margins} in Appendix \ref{app8}). The following property can be proved using the same arguments
as in Case 1 in the proof of Lemma \ref{lemma2}.

\begin{observation}
\label{obs2}
For $(\sigma, t, u) \in S_t$, $u$ is matched at $i$ in $\sigma_u^i$ if and only if $i < t$.
\end{observation}

\begin{definition}[Expected Marginal Loss $\alpha_t$]
\begin{equation}
\label{eq1}
\mbox{Expected marginal loss at position $t$}\ =\ \alpha_t \ =\ \frac{\sum_{(\sigma, t, u)\in S_t}b_u}{k^n}
\end{equation}
\end{definition}

\begin{claim}
\label{claim2}

\begin{equation}
\label{eq7}
\mbox{$\forall\ t$:}\ \ \ \ \ x_t\ =\ B - \sum_{s\leq t}\alpha_s
\end{equation}
\begin{equation}
\label{eq6}
\mbox{Total loss}\ =\ \sum_t{(B-x_t)}\ =\ \sum_t{(k-t+1)\alpha_t}
\end{equation}
\end{claim}
\begin{proof}
To prove equation \eqref{eq7}, we will fix a $t$ and construct a
one-to-one map $g: R_t \rightarrow \bigcup_{s\leq t}S_t$. Given
$(\sigma, t, u) \in R_t$, let $i$ be the lowest position of $u$ such
that $u$ remains unmatched in $\sigma_u^i$. By observation \ref{obs2},
$i$ is unique for $(\sigma, t, u)$. We let $g(\sigma, t, u) =
(\sigma_u^i, i, u)$. Clearly, $(\sigma_u^i, i, u) \in S_i$. To prove that the map is one-to-one, suppose $(\rho, s, u) = g(\sigma, t, u) =
g(\overline{\sigma}, t, u)$. Then by definition of $g$, $\rho =
\sigma_u^s = \overline{\sigma}_u^s$ which is only possible if $\sigma
= \overline{\sigma}$. Therefore, $|R_t| = \bigcup_{s\leq t}S_t$.

Lastly, observe that $g$ maps an element of $R_t$ corresponding to the
vertex $u$ being unmatched, to an element of $S_i$ corresponding to
the same vertex $u$ being unmatched. From equation \eqref{eq8},

$$B-x_t \ = \ \frac{\sum_{(\sigma, t, u)\in R_t}{b_u}}{k^n} \ = \ \sum_{i\leq t}\frac{\sum_{(\sigma_u^i, i, u)\in S_i}{b_u}}{k^n} \ = \ \sum_{i\leq t}\alpha_i$$

This proves equation \eqref{eq7}. Summing \eqref{eq7} for all $t$, we get \eqref{eq6}.
\end{proof}

Now consider the same set-valued map $f$ from Definition \ref{def:map}, but restricted only to the members of $\bigcup_tS_t$. We have:

\begin{claim}
\label{claim3}
For $(\sigma, t, u) \in S_t$ and $(\overline{\sigma}, \overline{t},
\overline{u}) \in S_{\overline{t}}$, if $(\rho, s, u') \in f(\sigma,
t, u)$ and $(\rho, s, u') \in f(\overline{\sigma}, \overline{t},
\overline{u})$ then $\sigma =\overline{\sigma}$, $t = \overline{t}$
and $u = \overline{u}$.
\end{claim}
\begin{proof}
If $u'$ is matched to $v$ in $\rho$ then by definition of $f$, $v =
u^* = \overline{u}^*$, implying $u = \overline{u}$. Therefore, $\rho =
\sigma_u^i = \overline{\sigma}_u^i$ for some $i$. But this implies
that $\overline{\sigma} = \sigma_u^j$ for some $j$. This is only
possible for $j = t$ since by definition, if $u$ is unmatched in
$\sigma$ at $t$, then there exists a unique $i$ for which
$(\sigma_u^i, i, u) \in \bigcup_t{S_t}$. If $j = t$, then $\sigma =
\overline{\sigma}$ and $t = \overline{t}$.
\end{proof}

Armed with this disjointness property, we can now prove our main theorem. 

\begin{theorem}
As $k \rightarrow \infty$,
\begin{equation}
\sum_tx_t\ \geq \ \left(1-\frac{1}{e}\right)\mathrm{OPT}(G)
\end{equation}
\end{theorem}
\begin{proof}
Using Lemma~\ref{lemma2} and Observation~\ref{obsk}, we have for every
$(\sigma, t, u) \in S_t$,

\begin{equation}
\label{eq10}
\psi(t)b_u\ \leq\ \frac{1}{k}\displaystyle\sum_{(\sigma_u^i, s, u')\in f(\sigma, t, u)}\psi(s)b_{u'}
\end{equation}

If we add the equation \eqref{eq10} for all $(\sigma, t, u) \in S_t$ and for all $1\leq t\leq n$, then using Claim \ref{claim3} and Observation \ref{obs1}, we arrive at

\begin{eqnarray}
\nonumber \sum_t\psi(t)\frac{\sum_{(\sigma, t, u)\in S_t}b_u}{k^n} & \leq & \frac{1}{k}\sum_t\psi(t)\frac{\sum_{(\sigma, t, u)\in Q_t}{b_u}}{k^n}\\
\label{eq14} \sum_t\psi(t)\alpha_t & \leq & \frac{1}{k}\sum_t\psi(t)x_t\\
\label{eq15} & = & \frac{1}{k}\sum_t\psi(t)\left(B - \sum_{s\leq t}\alpha_s\right)
\end{eqnarray}

Equation \eqref{eq14} follows from \eqref{eq1} and
\eqref{eq9}. Equation \eqref{eq15} uses Claim \ref{claim2}.

We now rearrange terms to get 
\begin{equation}
\label{eq18}
\sum_t\alpha_t\left(\psi(t) + \frac{\sum_{s\geq t}\psi(s)}{k}\right) \ \leq \ \frac{B}{k}\sum_{t}\psi(t)
\end{equation}

When $\psi(t) = 1-\left(1-\frac{1}{k}\right)^{k-t+1}$, observe that
$\psi(t) + \frac{\sum_{s\geq t}\psi(s)}{k} \geq \frac{(k-t+1)}{k}$ and
$\sum_{t}\psi(t) = \frac{k}{e}$ as $k \rightarrow \infty$. Using
Claim~\ref{claim2},

\begin{eqnarray}
\nonumber \mbox{Total loss} & = & \sum_t{(B-x_t)} \ = \ \sum_{t}(k-t+1)\alpha_t\\
\nonumber & \leq & k\sum_t\alpha_t\left(\psi(t) + \frac{\sum_{s\geq t}\psi(s)}{k}\right)\\
\nonumber & \leq & B\sum_{t}\psi(t)\\
\nonumber & = & \frac{kB}{e} ~~~~~~~~~~~~~ \mbox{as} ~ k \rightarrow \infty\\
\nonumber & = & \frac{\mathrm{OPT}(G)}{e}
\end{eqnarray}

Hence, as $k \rightarrow \infty$, $$\mbox{Total gain}\ \geq\ \left(1-\frac{1}{e}\right)\mathrm{OPT}(G)$$ 

\noindent \textbf{Remark}: Observe that we substituted for $\psi(t)$ only after equation \eqref{eq18} - up until that point, any choice of a non-increasing function $\psi$ would have carried the analysis through. In fact, the chosen form of $\psi$ is a result of trying to reduce the left hand side of equation \eqref{eq18} to the expected total loss. To conclude, the `right' perturbation function is dictated by the analysis and not vice versa.

\end{proof}

\bibliographystyle{alpha}
\bibliography{weighted_matching}

\appendix

\section{The Reduction from Online Budgeted Allocation with Single Bids}
\label{app7}

In this section, we will show that the \emph{single bids} case of the online budgeted allocation problem reduces to online vertex-weighted bipartite matching. Let us first define these problems.\\

\noindent \textbf{\textsc{Online budgeted allocation}}: We have $n$ agents and $m$ items. Each agent $i$ specifies a monetary budget $B_i$ and a bid $b_{ij}$ for each item $j$. Items arrive online, and must be immediately allocated to an agent. If a set $S$ of items is allocated to agent $i$, then the agent pays the minimum of $B_i$ and  $\sum_{j\in S}b_{ij}$. The objective is to maximize the total revenue of the algorithm.\\

\noindent \textbf{\textsc{Single bids case}}: Any bid made by agent $i$ can take only two values: $b_i$ or 0. In other words, all the non-zero bids of an agent are equal.

\begin{claim}
Online budgeted allocation with single bids reduces to online vertex-weighted bipartite matching.
\end{claim}
\begin{proof}
Given an instance of online budgeted allocation where agent $i$ has budget $B_i$ and single bid value $b_i$, we will construct an input instance $G(U, V, E, \{b_u\}_{u\in U})$ of online vertex-weighted bipartite matching. The set $V$ consists of one vertex corresponding to every item. The set $U$ will contain one or more vertices for every agent.

For every agent $i$, let $n_i$ be the largest integer such that $n_ib_i \leq B_i$ and let $r_i = B_i - n_ib_i$. Clearly, $r_i < b_i$. We will construct a set $U_i$ of $n_i$ vertices, each with weight $b_i$. In addition, if $r_i > 0$, then we will construct a vertex $\bar{u}_i$ with weight $r_i$ and add it to $U_i$. For all $u \in U_i$ and $v \in V$, the edge $uv \in E$ if and only if agent $i$ makes a non-zero bid on the item corresponding to $v$.\\

\noindent (1) Given a solution to the budgeted allocation problem where a set $S_i$ of items is allocated to agent $i$, let us see how to construct a solution to the vertex-weighted matching problem with the same total value. 
\begin{itemize}
\item If agent $i$ pays a total of $|S_i|\cdot b_i$, then we know that $|S_i| \leq n_i$. Hence, for every item in $S_i$, we will match the corresponding vertex in $V$ to a vertex in $U_i - \{\bar{u}_i\}$. Let $R_i$ be the set of vertices in $U_i$ thus matched. We have: $$\sum_{u \in R_i}b_u\ =\ |R_i|\cdot b_i\ =\ |S_i|\cdot b_i$$

\item If agent $i$ pays a total amount strictly less than $|S_i|\cdot b_i$, then we know that: (a) $|S_i| \geq n_i+1$, (b) $r_i > 0$ and (3) agent $i$ pays the budget $B_i$. We can now choose any $n_i+1$ items in $S_i$ and match the corresponding vertices in $V$ to the $n_i+1$ vertices in $U_i$. The sum of the weights of matched vertices in $U_i$, $\sum_{u \in U_i}b_u = B_i$.
\end{itemize}
Summing over all $i$, the weight of the matching formed is equal to the total revenue of the budgeted allocation. Let $\mathrm{OPT}_A$ and $\mathrm{OPT}_M$ denote the values of the optimal solutions of the budgeted allocation and the vertex-weighted matching problems respectively. Then we conclude from the above discussion that: 
\begin{equation}
\label{eq19}
\mathrm{OPT}_M\ \geq \mathrm{OPT}_A
\end{equation}

\noindent (2) Given a solution to the vertex-weighted matching problem where a set $R \subseteq U$ of vertices is matched, let us see how to construct a solution to the budgeted allocation problem with at least the same total value. Let $R_i = R\cap U_i$. For every $v \in V$ that is matched to a vertex in $R_i$, we will allocate the corresponding item to agent $i$. Let $S_i$ be the set of items allocated to agent $i$. 
\begin{itemize}
\item If $|R_i| = |S_i| \leq n_i$, then agent $i$ pays a total of $|S_i|\cdot b_i$ and we have: $$\sum_{u \in R_i}b_u\ \leq \ |R_i|\cdot b_i\ =\ |S_i|\cdot b_i$$ 
\item If on the other hand, $|R_i| = |S_i| = n_i+1$ then agent $i$ pays a total of $B_i$ and we have: $$\sum_{u \in R_i}b_u\ =\ \sum_{u \in U_i}b_u\ =\ B_i$$
\end{itemize}
Summing over all $i$, the total revenue of the budgeted allocation is at least the weight of the matching. Let $\mathrm{ALG}_M$ be the expected weight of the vertex-weighted matching constructed by \pgreedy~and $\mathrm{ALG}_A$ be the expected value of the budgeted allocation constructed using the above scheme. From the above discussion, we conclude:
Therefore, 
\begin{eqnarray}
\nonumber \mathrm{ALG}_A & \geq & \mathrm{ALG}_M\\
\label{eq20} & \geq & \left(1-\frac{1}{e}\right)\mathrm{OPT}_M\\
\nonumber & \geq & \left(1-\frac{1}{e}\right)\mathrm{OPT}_A
\end{eqnarray} 

Here, equation \eqref{eq20} follows from the main result - Theorem \ref{thm:main} - and the last step uses equation \eqref{eq19}. This completes our proof.

\end{proof}

\section{Performance of {\sc Greedy} and {\sc Ranking}}
\label{app3}
With non-equal weights, it is clearly preferable to match vertices with larger weight. This
leads to the following natural algorithm.\\

\begin{algorithm}[H]
\caption{{\sc Greedy}}
\ForEach{arriving $v \in V$}
{
	Match $v$ to the unmatched neighbor in $u$ which maximizes $b_u$ (breaking ties arbitrarily)\;
}
\end{algorithm}

It is not hard to show that {\sc Greedy} achieves a competitive ratio
of at least $\frac{1}{2}$. 

\begin{lemma}
{\sc Greedy} achieves a competitive ratio of 1/2 in vertex-weighted online bipartite matching.
\end{lemma}
\begin{proof}
Consider an optimal offline matching, and a vertex $u\in U$ that is
matched in the optimal offline matching but not in the greedy
algorithm. Now look at a vertex $u^*\in V$ that is matched to the
vertex $u$ in the optimal matching.  In \greedy, $u^*$ must have been
matched to a vertex $u^{'} \in U$, s.t. $b_u \leq b_{u^{'}}$, since $u$ was
unmatched when $u^*$ was being matched. So we'll charge the loss of
$b_u$ to $u^{'}$. Note that each $u^{'}$ does not get charged more than
once -- it is charged only by the optimal partner of its partner in
the algorithm's matching. Thus the loss of the algorithm is no more
than the value of the matching output by the algorithm. Hence the
claim.
\end{proof}

In fact, this factor $1/2$ is tight for {\sc Greedy} as shown by an instance consisting
  of many copies of the following gadget on four vertices, with $u_1,
  u_2 \in U$ and $v_1, v_2 \in V$. As $\epsilon \rightarrow 0$, the
  competitive ratio of {\sc Greedy} tends to $\frac{1}{2}$.

\begin{center}
\includegraphics[height=3cm]{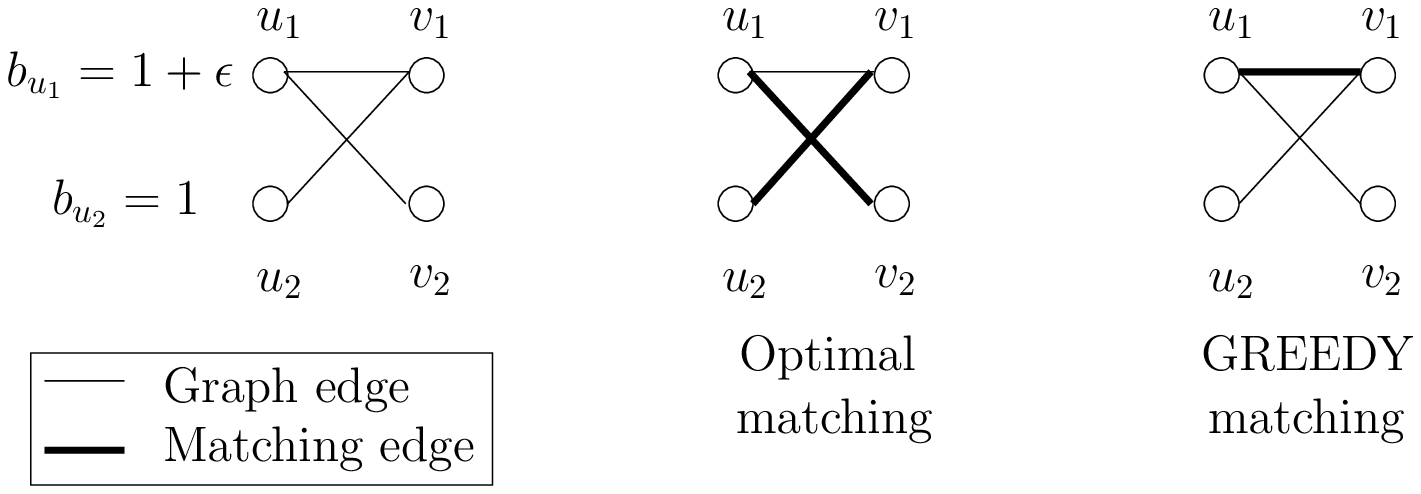}
\end{center}

Notice that this counter-example relies on weights being roughly
equal. We, however, know that {\sc Ranking} has an expected
competitive ratio of $(1-1/e)$ when the weights are equal. On the
other hand, if the weights are very different, i.e. $\epsilon$ is
large, in the above example, then {\sc Greedy} provides a good
competitive ratio. At the same time, if we exchanged the weights on
the two vertices in the example to be $b_{u_1} = 1$ and $b_{u_2} =
1+\epsilon$, then as $\epsilon$ grows large, the expected competitive
ratio of {\sc Ranking} drops to $\frac{1}{2}$ and on larger examples,
it can be as low as $\frac{1}{n}$. To summarize, {\sc Greedy} tends to
perform well when the weights are highly skewed and {\sc Ranking}
performs well when the weights are roughly equal.

\section{Intuition Behind the Sufficiency of Independent Perturbations}
\label{app6}

Recall that our algorithm perturbs each weight $b_u$ independent of the other weights. The fact that \pgreedy~achieves the best possible competitive ratio is a post-facto proof that such independence in perturbations is sufficient. Without the knowledge of our algorithm, one could reasonably believe that the vector of vertex-weights $\{b_u\}_{u\in U}$ - which is known offline - contains valuable information which can be exploited. In what follows we provide intuition as to why this is not the case.

Consider the two input instances in Figure \ref{fig:fig4}. Both the connected components in $G_1$ have equal weights, and hence we know that \ranking~achieves the best possible competitive ratio on $G_1$. Similarly, both connected components in $G_2$ have highly skewed weights, suggesting \greedy~as the optimal algorithm. On the other hand, \ranking~and \greedy~are far from optimal on $G_2$ and $G_1$ respectively. Since two instances with identical values of vertex-weights require widely differing strategies, this exercise suggests that we may not be losing must information by perturbing weights independently. The optimality of our algorithm proves this suggestion.

\begin{figure}[h]
\begin{center}
\includegraphics[height=6cm]{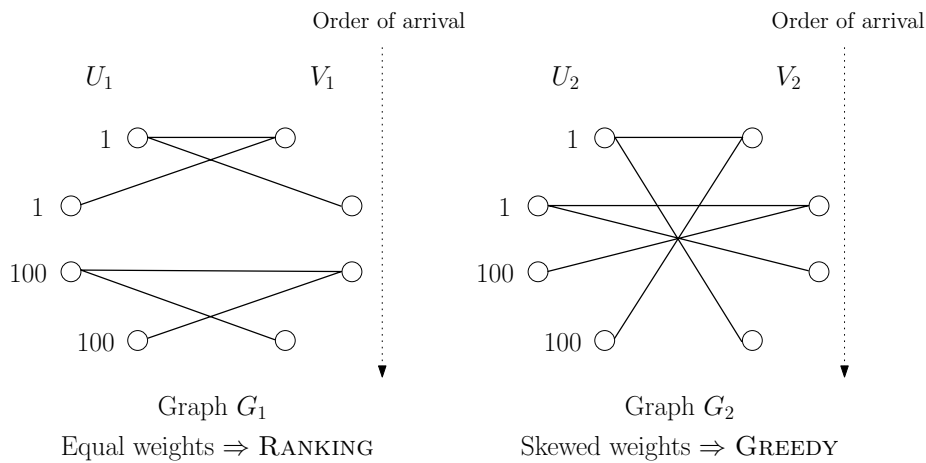}
\end{center}
\caption{\label{fig:fig4}Two instances with the same vertex-weights, but widely differing optimal strategies.}
\end{figure}

\section{Hard Instances in $2\times 2$ Graphs}
\label{app9}

Figure \ref{fig:twobytwo} shows the only two potentially `hard' instances in $2\times 2$ graphs. On all other instances, the optimal matching is found by any reasonable algorithm that leaves a vertex $v \in V$ unmatched only if all its neighbors are already matched.

\begin{figure}[h]
\begin{center}
\includegraphics[height=2cm]{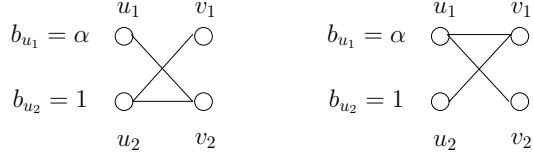}
\end{center}
\caption{\label{fig:twobytwo}Canonical examples for 2$\times$2 graphs.}
\end{figure}

\newpage

\section{Marginal Loss Events}
\label{app8}

\begin{figure}[h]
\begin{center}
\includegraphics[height=3.5cm]{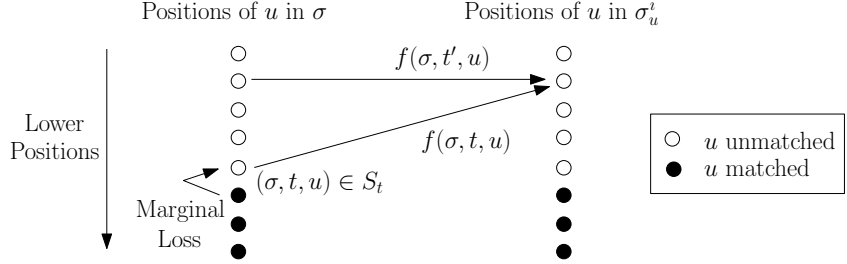}
\end{center}
\caption{\label{fig:margins}Marginal Losses}
\end{figure}

\section{Graphs with Imperfect Matchings}
\label{app1}
In Section \ref{sec:main-proof}, we proved Theorem \ref{thm:main} for graphs $G(U, V, E)$ such that $|U| = |V|$ and $G$ has a perfect matching. We can remove these assumptions with just a few modifications to the definitions and equations involved in the proof. The algorithm remains unchanged, \textit{i.e.} we just use \pgreedy.  We will only outline these modifications and the rest of the proof follows easily. Let $M^*(G)$ be a maximum weight matching in $G(U, V, E)$ and $\overline{U}$ be the set of vertices in $U$ matched by $M^*(G)$. Thus we know that $\mathrm{OPT}(G) = \sum_{u \in \overline{U}}b_u$.

Keeping the definition of $Q_t$ the same, we change the definition of $R_t$ to:
$$R_t = \{(\sigma, t, u)\ :\ \mbox{$u\in \overline{U}$ and $\sigma(u) = t$ and $y_{\sigma,t} = 0$}\}$$

The above redefinition conveys the fact that if a vertex $u$ is \emph{not} matched by $M^*(G)$, then we no longer consider $u$ being unmatched a bad event. Consequently, equation \eqref{eq8} changes to:

$$B-x_t \leq \frac{\sum_{(\sigma, t, u)\in R_t}{b_u}}{k^n}$$ which in turn yields following counterpart of equation \eqref{eq7}:

\begin{equation}
\label{eq16}
\forall t,\ \ \ \ \ \ x_t\ \geq\ B - \sum_{s\leq t}\alpha_s
\end{equation}

Let $\mathrm{Eq}(t)$ be the version of \eqref{eq16} for $t$. We then multiply $\mathrm{Eq}(t)$ by $\psi(t) - \psi(t+1)$ and sum over $1 \leq t \leq n$ to obtain a combined inequality (with $\psi(k+1) = 0$):

\begin{eqnarray}
\nonumber \sum_t{\left(\psi(t)\ -\ \psi(t+1)\right)x_t} & \geq & \psi(1)B\ -\  \sum_t{\psi(t)\alpha_t}\\
\label{eq17} \sum_t{\psi(t)\alpha_t} & \geq & \psi(1)\frac{\mathrm{OPT}(G)}{k} \ -\ \sum_t{\frac{\left(1-\psi(t+1)\right)}{k}x_t}
\end{eqnarray}

Equation \eqref{eq17} used the definition of $\psi(t) = 1-\left(1-\frac{1}{k}\right)^{(k-t+1)}$. Combining equation \eqref{eq17} with \eqref{eq14}, we get:

\begin{eqnarray}
\nonumber \frac{1}{k}\sum_t{\psi(t)x_t} & \geq & \psi(1)\frac{\mathrm{OPT}(G)}{k} - \sum_t{\frac{\left(1-\psi(t+1)\right)}{k}x_t}\\
\nonumber \sum_t{x_t} & \geq & \psi(1)\mathrm{OPT}(G) - \sum_t{\left(\psi(t)-\psi(t+1)\right)x_t}\\
\nonumber & \geq & \left(1-\frac{1}{e}\right)\mathrm{OPT}(G)
\end{eqnarray}

as $k \rightarrow \infty$, since $\psi(1) \rightarrow \left(1-\frac{1}{e}\right)$ and $\psi(t) - \psi(t+1) = \frac{(1-\psi(t+1)}{k} \rightarrow 0$ as $k \rightarrow \infty$.

\section{A Lower Bound for Randomized Algorithms with Edge Weights}
\label{app4}

In this section, we will sketch the proof of a lower bound for the competitive ratio of a randomized algorithm, when the graph $G(U,V,E)$ has edge weights and our objective is to find a matching in $G$ with maximum total weight of edges. Previous studies of this problem have only mentioned that no constant factor can be achieved when the vertices in $V$ arrive in an online manner. However, we have not been able to find a proof of this lower bound for randomized algorithms in any literature. We prove the result when the algorithm is restricted to be scale-free. A scale-free algorithm in this context produces the exact same matching when all the edge weights are multiplied by the same factor.

Consider a graph $G(U, V, E)$ such that $U$ contains just one vertex $u$ and each vertex in $v \in V$ has an edge to $u$ of weight $b_v$. Fix $v_1, v_2, ...$ to be the order in which the vertices of $V$ arrive online. By Yao's principle, it suffices for us to produce a probability distribution over $b_{v_1}, b_{v_2}, ...$ such that no deterministic algorithm can perform well in expectation. We will denote the vector of edge weights in the same order in which the corresponding vertices in $V$ arrive, \textit{i.e.} $(b_{v_1}, b_{v_2}, ...)$ and so on. Consider the following $n$ vectors of edge weights: For every $1 \leq i \leq n$, $\mathbf{b}_i = (D^i, D^{i+1}, ..., D^n, 0, 0, ...)$ and so on, where $D$ is a sufficiently large number. Suppose our input distribution chooses each one of these $n$ vectors of edge weights with equal probability.

Clearly, regardless of the vector which is chosen, $\mathrm{OPT}(G) = D^n$. Since an algorithm is assumed to be scale-free and online, it makes the exact same decisions after the arrival of first $k$ vertices for each of the edge weight vectors $\mathbf{b}_j$, $1 \leq j \leq k$. Therefore, it cannot distinguish between $\mathbf{b}_1, ..., \mathbf{b}_k$ after just $k$ steps. Hence, we can characterize any algorithm by the unique $k$ such that it matches the $k$'th vertex in $V$ with a positive weight edge.

Let $\mathrm{ALG}$ be any deterministic algorithm that matches the $k$'th incoming vertex with a positive weight edge to $u$. Then the expected weight of the edge chosen by $\mathrm{ALG}$ is $\displaystyle\frac{1}{n}\sum_{i> k}{D^i}$. Since $D$ is large, this is at most $\frac{c}{n}\mathrm{OPT}(G)$, where $c$ is some constant. Applying Yao's principle, we conclude that the competitive ratio of the best scale-free randomized algorithm for online bipartite matching with edge weights is $O\left(\frac{1}{n}\right)$. 

\end{document}